\documentclass[journal,letterpaper,twocolumn,twoside,nofonttune]{IEEEtran}
\usepackage{times}

\usepackage{amsmath}
\usepackage{amsfonts}
\usepackage{latexsym}
\usepackage{amssymb}
\usepackage{algorithm2e}
\usepackage{graphics}
\usepackage{graphicx}
\usepackage{subfigure}

\usepackage{mathrsfs} 

\usepackage{cite}  
\usepackage{upref}

\usepackage{theorem}
\usepackage{psfrag}

\title{\Huge$\,$\\[-2.75ex]
{List-decoding of Subspace Codes and Rank-Metric Codes up to Singleton Bound}\\[0.50ex]}

\author{%
\authorblockN{\large{\bf%
Hessam Mahdavifar and Alexander Vardy}\vspace{0.25ex}}\\
\authorblockA{\large
University of California San Diego, La Jolla, CA\,92093, USA\\ 
{\tt \{hessam@ucsd.edu,\,avardy@ucsd.edu\}}}\vspace*{-1.5ex}
}

\renewcommand{\markboth}[2]
{\renewcommand{\leftmark}{#1}\renewcommand{\rightmark}{#2}}

\theoremstyle{plain} 
\theorembodyfont{\normalfont\slshape}

\newtheorem{thm}{Theorem\hspace{-1pt}} 
\newenvironment{theorem}
{\begin{thm}\hspace*{-1ex}{\bf.}}{\end{thm}}

\newtheorem{lem}[thm]{Lemma\hspace{-.75pt}}
\newenvironment{lemma}{\begin{lem}\hspace*{-1ex}{\bf.}}{\end{lem}}

\newtheorem{prop}[thm]{Proposition$\!$}

\newtheorem{cor}[thm]{Corollary$\!$}
\newenvironment{corollary}{\begin{cor}\hspace*{-1ex}{\bf.}}{\end{cor}}

\newtheorem{defn}{Definition$\!$}

\newcounter{enumrom}
\renewcommand{\theenumrom}{(\roman{enumrom})}

\makeatletter
\renewcommand{\@endtheorem}{\endtrivlist}
\makeatother

\makeatletter
\renewcommand{\thefigure}{{\bf \@arabic\c@figure}}
\renewcommand{\fnum@figure}{{\bf Figure}\,\thefigure}
\makeatother

\renewcommand{\QEDclosed}{\mbox{\rule[-1pt]{1.3ex}{1.3ex}}} %

\newcommand{\be}[1]{\begin{equation}\label{#1}}
\newcommand{\ee}{\end{equation}} 
\newcommand{\eq}[1]{(\ref{#1})}

\renewcommand{\leq}{\leqslant}
 
\renewcommand{\geq}{\geqslant}

\newcommand{\Tref}[1]{Theo\-rem\,\ref{#1}}

\newcommand{\Lref}[1]{Lem\-ma\,\ref{#1}}
\newcommand{\Cref}[1]{Co\-ro\-lla\-ry\,\ref{#1}}

\newcommand{\deff}{\mbox{$\stackrel{\rm def}{=}$}}

\newcommand{\al}{\alpha}

\newcommand{\gam}{\gamma}
\newcommand{\eps}{\epsilon}

\newcommand{\cC}{{\cal C}}

\newcommand{\cE}{{\cal E}}

\newcommand{\cG}{{\cal G}}

\newcommand{\cP}{{\cal P}}
\newcommand{\cQ}{{\cal Q}}

\DeclareMathAlphabet{\mathbfsl}{OT1}{ppl}{b}{it}

\newcommand{\bU}{\mathbfsl{U}} 
\newcommand{\bV}{\mathbfsl{V}}
 
\newcommand{\bX}{\mathbfsl{X}}
\newcommand{\bY}{\mathbfsl{Y}}

\newcommand{\bu}{\mathbfsl{u}}

\renewcommand{\Bbb}{\mathbb}

\newcommand{\F}{\Bbb{F}}
\newcommand{\Fq}{\Bbb{F}_{\!q}}

\newcommand{\Fqm}{\Bbb{F}_{q^m}}

\newcommand{\script}[1]{{\mathscr #1}}

\newcommand{\nsfrac}[2]{\mbox{\raisebox{.8mm}{\normalsize $\scriptstyle #1$} 
                          \normalsize$\!\! \kern-1pt / \kern-1pt \!\!$
                          \raisebox{-.8mm}{\normalsize $\scriptstyle #2$}}}

\newcommand{\lsfrac}[2]{\mbox{\raisebox{.8mm}{\large $\scriptstyle #1$} 
                          \large$\!\! \kern-1pt / \kern-1pt \!\!$
                          \raisebox{-.8mm}{\large $\scriptstyle #2$}}}

\renewcommand{\L}{\script{L}}

\newcommand{\fbu}{f_{\bu}}

\newcommand{\Rs}{R^{*}}

\begin{document}

\maketitle
\thispagestyle{empty}

\begin{abstract}
Subspace codes and rank-metric codes can be used to correct errors and erasures in network, with linear network coding. Both types of codes have been extensively studied in the past five years.

Subspace codes were introduced by Koetter and Kschischang to correct errors and erasures in networks where topology is unknown (the noncoherent case). In this model, the codewords are vector subspaces of a fixed ambient space; thus codes for this model are collections of such subspaces. In a previous work, we have developed a family of subspace codes, based upon the Koetter-Kschichang construction, which are efficiently list decodable. Using these codes, we achieved a better decoding radius than Koetter-Kschischiang codes at low rates. Herein, we introduce a new family of subspace codes based upon a different approach which leads to a linear-algebraic list-decoding algorithm. The resulting error correction radius can be expressed as follows: for any integer $s$, our list-decoder using $s+1$-interpolation polynomials guarantees successful recovery of the message subspace provided the normalized dimension of errors is at most $s(1-sR)$. The same list-decoding algorithm can be used to correct erasures as well as errors. The size of output list is at most $\cQ^{s-1}$, where $\cQ$ is the size of the field that message symbols are chosen from. 

Rank-metric codes are suitable for error correction in the case where the network topology and the underlying network code are known (the coherent case). Gabidulin codes are a well-known class of algebraic rank-metric codes that meet the Singleton bound on the minimum rank metric of a code. In this paper, we introduce a folded version of Gabidulin codes analogous to the folded Reed-Solomon codes of Guruswami and Rudra along with a list-decoding algorithm for such codes. Our list-decoding algorithm makes it possible to recover the message provided that the normalized rank of error is at most $1-R-\eps$, for any $\eps > 0$. Notably this achieves the information theoretic bound on the decoding radius of a rank-metric code. 
\end{abstract}

\begin{keywords} 
list-decoding, subspace codes, rank-metric codes, linearized polynomial, Gabidulin codes
\vspace{-1.00ex}
\end{keywords}

\section{Introduction} 
\label{sec:Introduction}
\noindent\looseness=-1

Subspace codes and rank metric codes are two closely related family of codes used for reliable communication of messages in linear network coding \cite{KK} \cite{SKK}. Network coding, in general, is a technique where nodes of the network take several packets and combine them together for transmission instead of simply relaying the packets of information they receive \cite{ACLY}. It is known that in multicast setting, where one transmitter communicates with several receivers in the network simultaneously, linear network coding, wherein all the operations performed at network nodes are linear operations, is sufficient to achieve the individual max-flow bound on the rate of communication between the transmitter and each of the receivers \cite{LYC}. 

In random linear network coding, each intermediate node of the network creates a random linear combination of the packets it receives and sends it through its output links \cite{HKMKE}. It is proved that random linear network coding is as good as linear network coding, in the context of multicast setting, with high probability if the size of the field the message symbols are chosen from is large enough \cite{HMKKESL}. Subspace codes have been recently introduced in order to enable reliable communication of messages in random linear network coding \cite{KK}. Let the ambient space $W$ be a vector space over a finite field $\Fq$. A subspace code in $W$ is a non-empty subset of all the subspaces of $W$. Koetter-Kschischang algebraic construction of subspace codes, originally called Reed-Solomon-like codes in \cite{KK},  is analogous to Reed-Solomon codes in classical block codes wherein symbols are replaced by vectors, regular polynomials with \emph{linearized polynomials}, and sequences of symbols with $\Fq$-linear span of the corresponding vectors.

In a previous work, we proposed a new family of subspace codes that enables list-decoding, hence achieving a better tradeoff between rate and error correction capability \cite{MV1}, \cite{MV2}. The idea was to evaluate all the powers of the linearized message polynomial, up to some power $L$, in order to list-decode with list size at most $L$.  In a sense our algorithm can be regarded as an analogous to Sudan list-decoding algorithm of Reed-Solomon codes \cite{S}. We further improved this result by introducing multiplicity in the ring of linearized polynomials \cite{MV3}. This was motivated by Guruswami-Sudan list-decoding algorithm of Reed-Solomon codes \cite{GS}. We were able to list-decode for a wider range of rates and also to achieve a better tradeoff between the rate and decoding radius by enforcing multiple roots for the interpolation polynomial \cite{MV3}.

In this paper, we introduce a new family of subspace codes that allows a simple linear-algebraic list-decoding by using $s+1$-variate interpolation polynomials, where $s$ is a design parameter. The entire list-decoding algorithm is linear-algebraic. A system of linear equations is solved for the interpolation step and another linear system is solved to compute the set of all the possible solutions which indeed is a linear space. This is motivated by the recent work of Vadhan \cite[Ch. 5]{V} and Guruswami \cite{G2} which suggested a simplified version, with no need of multiplicity, of previously proposed list-decoding algorithm of folded Reed-Solomon codes by  Guruswami and Rudra in \cite{GR}. The later was built upon the work of Parvaresh and Vardy on list-decoding of Reed-Solomon codes by proposing multivariate interpolation \cite{PV}.  

In the coherent system of network coding, the network topology and the particular network coding operations done at intermediate nodes are known to both the transmitter and the receiver. In this setting, rank-metric codes turned out to be the suitable tool to deal with possible injected errors into the network \cite{SKK}. Each codeword in a rank-metric code is a matrix with fixed dimensions whose entries are taken from a finite field $\Fq$. The distance between two matrices is simply the rank of their difference. Gabidulin codes were introduced as a class of MRD (maximum rank distance) codes. They achieve the Singleton bound on the minimum rank distance of a rank metric code. In this paper, we define a folded version of Gabidulin codes. Then we propose a list-decoding algorithm that can correct the fraction of errors up to the Singleton bound which is the information theoretic upper bound on the error correction capability of a code.

The rest of this paper is organized as follows. We start with a brief overview of linearized polynomials, subspace codes, rank-metric codes and Gabidulin codes in Section \ref{sec:two}. In Section \ref{sec:three}, we first discuss our new construction of subspace codes, then we propose a list-decoding algorithm. After that, we establish the correctness of the algorithm and provide the decoding radius and other parameters of our code. In Section \ref{sec:four}, we introduce the folded version of Gabidulin codes and provide the list-decoding algorithm. Then we show that we are able to correct the fraction of errors up to $1-R$, $R$ being the rate of the code, hence achieving the Singleton upper bound on the error correction capability of rank-metric codes.
 
\vspace{1ex}
\section{Background and Prior Work}
\vspace{.25ex}
\label{sec:two}

In this section, we first explain the ring of linearized polynomials. We establish some relevant terminology for subspace codes and explain the Koetter-Kschischang construction of subspace codes. Then we briefly recap the results of \cite{MV1} and \cite{MV2} which provide a new construction of subspace codes. In \cite{MV1} and \cite{MV2} we suitably modified and extended Koetter-Kschischang construction in many important respects in order to enable list-decoding. Then we briefly discuss the results of the follow-up work in \cite{MV3} wherein we introduced multiple roots for the interpolation polynomial. At the end, we briefly review rank-metric codes and Gabidulin codes as a class of maximum rank distance codes.

A polynomial over some extension field $\Fqm$ of $\Fq$ is called linearized if it has the following form:
$$
f(X)=\sum^{s}_{i=0} a_i X^{q^i}, 
$$  
where $a_i \in \Fqm$, for $i=0,1,\dots,s$. Assuming that $a_s \neq 0 $ we say that the polynomial $f(X)$ has $q$-degree $s$ which means that its actual degree is $q^s$. When $q$ is fixed under discussion, we will let $X^{\left[i\right]}$ denote $X^{q^i}$. The main property of linearized polynomials from which they receive their name is that they act as linear maps with respect to $\Fq$. The set of linearized polynomials forms a non-commutative ring under addition $+$ and composition operation $\otimes$. For any two linearized polynomials $f_1(X)$ and $f_2(X)$, the composition operation $f_1(X) \otimes f_2(X)$ is defined to be the composition $f_1(f_2(X))$ which is always a linearized polynomial. The ring of linearized polynomials over $\Fqm$ is denoted by $\L_{q^m}[X]$.

\subsection{Subspace Codes}

Let $W$ be a fixed $N$-dimensional vector space over $\mathbb{F}_q$ and $\cG(W)$ denote the set of all subspaces of $W$. For any $V \in \cG(W)$, the dimension of $V$ is denoted by $\dim(V)$.  For any $A, B \in \cG(W)$, the distance between $A$ and $B$ is defined as follows:
$$
d(A,B) \hspace*{6pt} \deff\ \hspace*{6pt} \dim(A+B)-\dim(A\cap B)
$$
$\cG(W)$ is indeed a metric space under this metric. Let $\cG(W,n)$ denote the set of all $n$-dimensional subspaces of $W$. A code $\mathcal{C}$ in the ambient space $W$ is a non-empty subset of $\cG(W)$. A codeword is an element of $\mathcal{C}$ which is in fact a subspace of $W$. 
\\\textbf{Definition.\,}\cite{KK} Let $\mathcal{C}$ be a code associated with the ambient space $W$ of dimension $N$ over $\mathbb{F}_q$. Suppose that the dimension of any $V \in \mathcal{C}$ is at most $n$. Then the rate of the code $R$ is defined as follows.
\be{symbol_rate}
R\hspace{6pt} {\deff} \hspace{6pt} \frac{\log_q\left|\mathcal{C}\right|}{nN}
\ee
In \cite{MV3}, we defined a new parameter, called the \emph{packet rate} of the code. The packet rate $\Rs$ is defined as follows:
\be{packet_rate}
\Rs \hspace{6pt} {\deff} \hspace{6pt} \frac{\log_{q^m}\left|\mathcal{C}\right|}{n} = 
\frac{\log_q\left|\mathcal{C}\right|}{nm}
\ee
where $q^m$ is the size of the underlying extension field. 
\hfill\raisebox{-0.5ex}{$\Box$}\vspace{1.0ex}

Koetter-Kschischang construction \cite{KK} of subspace codes can be regarded as an analogous to Reed-Solomon codes wherein symbols are replaced by vectors, polynomials with linearized polynomials and sequences of symbols with $\Fq$-linear span of the corresponding vectors.  Fix $m$ and an extension field $\Fqm$ of $\Fq$. $\Fqm$ can be also regarded as a vector space of dimension $m$ over $\Fq$. Fix a set $A=\left\{\alpha_1,\dots,\alpha_n\right\}$ of $n$ linearly independent vectors in $\Fqm$. Let $\bu=(u_0,\dots,u_{k-1})$ be the message vector and $\fbu(X)=\sum^{k-1}_{i=0} u_iX^{\left[i\right]}$ be the corresponding linearized message polynomial. Then the corresponding codeword $V$ is the $\Fq$-linear span of the set $\left\{(\alpha_i,f(\alpha_i)):1 \leq i \leq n\right\}$ which is an $n$-dimensional vector space. The ambient space $W$ is equal to $\left\langle A\right\rangle \oplus \Fq^m$ which is an $(n+m)$-dimensional vector space over $\Fq$. The codewords of Koetter-Kschischang code, which we simply call KK code, are in fact $n$ dimensional subspaces of the ambient space $W$. Each element of $W$ is represented as a vector $(x,y)$ where $x$ belongs to the span of $\al_i$'s and $y$ is an element of $\Fqm$.  

At the decoder, a nonzero bivariate linearized polynomial $Q(X,Y)$ of the form
$$
Q(X,Y)=Q_0(X)+Q_1(Y),\ 
$$
is constructed, where $Q_0$ and $Q_1$ are subject to some degree constraints such that $Q(x_i,y_i)=0$ for all the basis elements of the received subspace. Then the equation $Q(X,f(X)) = 0$ is solved to recover the message polynomial. It is proved in \cite{KK} that if not too many errors and erasures happen, then $\fbu(X)$ is the unique solution to this equation. Koetter and Kschischang give the normalized decoding radius of this scheme as

\be{tau-KK}
\frac{n-k+1}{n} \, =\, 1-\frac{k-1}{n} \, \approx \, 1-\Bigl(1+\frac{n}{m}\Bigr)R = 1 - \Rs
\ee

The main obstacle in the list-decoding of KK codes is that the ring of linearized polynomials is non-commutative. Because of that, an equation of certain degree over the ring of linearized polynomials may have exponentially many roots, while one has to guarantee a bounded list-size at the output of the decoder.  In order to enable list-decoding, we modified the KK construction in many important ways. Our work in \cite{MV1} and \cite{MV2} basically leads to a new construction of subspace codes which is list-decodable. 

Next, we turn to describe the encoding and decoding of this construction of subspace codes \cite{MV1}, \cite{MV2}. Recall from \cite[Ch. 4.9]{MS} that any finite extension $\mathbb{F}_{q^l}$ of $\mathbb{F}_q$ contains a primitive element $\gam$ such that $\gam,\gam^q,\dots,\gam^{q^{l-1}}$ forms a basis for $\F_{q^l}$ as a vector space over $\mathbb{F}_q$. This is called a normal basis for $\mathbb{F}_{q^l}$. Fix a finite field $\Fq$ and let $n$ divides $q-1$. Then the equation $x^n-1=0$ has $n$ distinct solutions in $\Fq$. Let $e_1=1,e_2,e_3,\dots,e_n$ be these solutions.  Let $\F=GF(q^{nm})$ and $\gamma$ be a generator of a normal basis for $\F$. Then define
\be{define-alpha}
\alpha_i=\gamma+e_i^{-1}\gamma^{q^m}+e_i^{-2}\gamma^{q^{2m}}+\dots+e_i^{-(n-1)}\gamma^{q^{(n-1)m}}
\ee
for $i=1,2,\dots, n$. For a given message polynomial $\fbu(X)$, our encoder constructs the vectors $v_i$'s as follows:
$$
v_i=(\alpha_i,\fbu(\alpha_i),\fbu^{\otimes 2}(\alpha_i),\dots, \fbu^{\otimes L}(\alpha_i))
$$
for $i=1,2,\dots,n$. Then it outputs the $n$-dimensional vector space spanned by $v_1,v_2,\dots,v_n$.  In this construction, the ambient space $W$ has dimension equal to $n+nmL$ and each element in $W$ is represented as a vector with $L+1$ coordinates such as $(x,y_1,y_2,\dots,y_L)$, where $x$ belongs to the vector space spanned by $\al_1,\al_2,\dots,\al_n$ and $y_i \in \F_{q^{nm}}$, for $i=1,2,\dots,L$. The decoding algorithm consists of three steps. In the first step, it computes the interpolation points. In the second step, a multivariate linearized polynomial $Q(X,Y_1,Y_2,\dots$ $,Y_L)$ of the form
$$
Q_0(X)+Q_1(Y_1)+Q_2(Y_2)+\dots+Q_L(Y_L)
$$
is constructed, where each $Q_i$ is subject to a degree constraint, such that  $Q(x,y_1,y_2,\dots,y_L)\, =\, 0$ for all the interpolation points $(x,y_1,y_2,\dots,y_L)$. Then in the factorization step, we compute all the roots $f(X) \in \L_q[X]$, with degree at most $k-1$, of the equation: 
$$
Q\bigl(X,f(X),\dots,f^{\otimes L}(X)\bigr) = 0
$$
To solve this equation efficiently, we propose a linearized version of Roth-Ruckenstein algorithm which was designed to solve equations over the ring of polynomials \cite{RR}. We also show in \cite{MV2} that there are at most $L$ solutions for $f(X) \in \L_q[X]$. Each solution corresponds to one possible output message.

We prove in \cite{MV2} that the normalized decoding radius of this list-decoding algorithm in terms of list size $L$ and packet rate $\Rs$ is given by
\be{MV-radius}
L \, - \, \frac{1}{2}L(L+1) \Rs
\ee 

We further improve this result by introducing multiplicities for the interpolation polynomial in \cite{MV3}. First, we establish the notion of multiplicity for linearized polynomials in this context. Then by enforcing multiple roots for the interpolation polynomial we achieve a better decoding radius. We are also able to list-decode at higher rates. For every positive integers $L$ and $c$, our list-$L$ decoder with multiplicity $c$ guarantees successful recovery of the message subspace provided that the normalized dimension of the error is at most
$$
\frac{2(L+1)}{c+1} \,-\, 1 \, - \,  \frac{L(L+1)}{c(c+1)}\Rs
$$
This improves the normalized decoding radius upon the previous results, given in \eq{tau-KK} and \eq{MV-radius}, for a wide range of rates. The parameter $c$ is independent of the code construction and can be chosen at the decoder in such a way that the decoding radius is maximized. As $L$ tends to infinity, the decoding radius of our construction with appropriate choice of $c$ approaches $\frac{1}{\Rs}-1$.

\subsection{Rank-Metric Codes}
\label{sec:two_B}

Rank-metric codes are suitable for the coherent system of network coding, where the network topology and the underlying network code are known to both the transmitter and the receiver \cite{SKK}. In \cite{SKK}, Silva et al. also show that subspace codes and rank-metric are closely related. Indeed, there is an injective mapping between rank-metric codes and subspace codes through a \emph{lifting} operation. 

Let $\Fq^{n \times m}$ denote the set of all $n \times m$ matrices over $\Fq$. For any $\bX \in \Fq ^{n \times m}$, let $\left\langle \bX \right\rangle$ denote the row space of the matrix $\bX$. A rank-metric code is just a subset of $\Fq ^{n \times m}$ which is called an array code in \cite{R}. The distance between $\bX , \bY \in \Fq ^{n \times m}$ is defined as rank$(\bX-\bY)$. We define the rate $R$ of a rank-metric code $\cC \subseteq \Fq ^{n \times m}$ as follows:
$$
R\hspace{6pt} {\deff} \hspace{6pt} \frac{\log_q(\left| \cC \right|)}{nm}
$$ 
The minimum (rank) distance of $\cC$ is the minimum distance between distinct elements of $\cC$. The Singleton bound is established in the context of rank-metric codes by Gabidulin in \cite{G}. It states that the minimum distance of a code $\cC$ with rate $R$, normalized by the number of rows $n$, is at most $1-R$. A rank-metric code that meets the Singleton bound on the minimum distance is called a maximum rank distance (MRD) code. Gabidulin codes are a class of MRD codes proposed in \cite{G}.

A Gabidulin code in $\Fq^{n \times m}$ is indeed a linear $(n, k)$ code over $\Fqm$ whose generator matrix $G$ has the following form:
$$
\left[
\begin{array}{cccc}
\al_1^{[0]} & \al_2^{[0]} & \dots & \al_n^{[0]} \\
\al_1 ^{[1]} & \al_2^{[1]} & \dots & \al_n^{[1]} \\
. & .  & .\,\,\, & . \\
. & .  & . & . \\
. & .  & \,\,\,. & . \\
\al_1^{[k-1]} & \al_2^{[k-1]} & \dots & \al_n ^{[k-1]}
\end{array}
\right]
$$
where the elements $\al_1, \al_2, \dots, \al_n \in \Fqm$ are linearly independent over $\Fq$. Each codeword is a column vector of length $n$ over $\Fqm$ which can be also regarded as a matrix in $\Fq^{n \times m}$. Note that the condition $n \leq m$ is required. The rate of the code is $R = \frac{k}{n}$. The minimum rank distance of a Gabidulin code is $d = n - k + 1$ which satisfies the Singleton bound in the rank metric \cite{G}. The minimum rank distance can be normalized to $1-R$. The unique decoding radius bound then becomes equal to $(1-R)/2$. A decoding algorithm which can correct errors, as long as the rank of error is less than $(d-1)/2$, is proposed in \cite{G}, hence achieving the bound $(1-R)/2$ on unique decoding radius.

Suppose that the input to the Gabidulin encoder is a message vector $\bu=\left[u_0\, u_1\,\dots\, u_{k-1}\right]$ which consists of $k$ message symbols in $\Fqm$. Let $\fbu(X)$ denote the corresponding linearized message polynomial $\sum^{k-1}_{i=0} u_iX^{\left[i\right]}$. Then the corresponding codeword $\bV = (\bu G)^T$ is indeed equal to 
$$
\bigl[\fbu(\al_1)\, \fbu(\al_2)\,\dots\, \fbu(\al_n)\bigr]^T
$$
which can be also regarded as a matrix in $\Fq ^{n \times m}$.

Now, the close relation between Koetter-Kschischang construction of subspace codes and Gabidulin codes becomes clear. Intuitively KK codes can be thought as a modification of Gabidulin codes where there is no ordering for the coordinates $\fbu(\al_i)$'s. Each $\al_i$ is appended to the corresponding $\fbu(\al_i)$, as a vector in $n$-dimensional vector space spanned by all $\al_i$'s, in order to keep track of evaluation points of the linearized polynomial $\fbu(X)$. More rigorously, the lifting mapping, defined in \cite{SKK}, translates Gabidulin codes into KK codes.  
\vspace{1ex}

\section{New Subspace Codes and Algebraic List-decoding Thereof}
\vspace{.25ex}
\label{sec:three}

In this section, we present a new construction of subspace codes and a list-decoding algorithm capable of correcting both errors and erasures. Our results in this section are motivated by the recent work of Vadhan \cite[Ch. p]{V} and Guruswami \cite{G2}. Then we establish the correctness of our algorithm and compute the error correction capability of the proposed construction. 

\subsection{Code Construction and List-decoding Algorithm}
\label{sec:three_A}

The following parameters of the construction are fixed: the finite field $\Fq$ and an extension $\Fqm$, the number of information symbols $k$, the dimension of code $n$ and  the parameter $s$ which is related to the list size. We require that $k \leq n \leq m$. A set $A=\left\{\al_1,\al_2,\dots,\al_n\right\}$ of linearly independent elements of $\Fqm$ is also fixed. In this construction, the ambient space $W$ is an $n+sm$-dimensional vector space over $\Fq$. Let $\gam$ be an element of $\Fqm$ which is not contained in any subfield of $\Fqm$ i.e. $\gam, \gam ^{q}, \dots, \gam^{q^{m-1}}$ are all distinct. 
\\\textbf{Encoding Algorithm:}
 \\Formally, the encoder is a function $\cE\! : \Fqm^k \! \to \cG(W,n)$. It accepts as input a message $\bu =(u_0,u_1,\dots,u_{k-1}) \in \Fqm^k$. The corresponding message polynomial is $\fbu (X) = \sum^{k-1}_{i=0} u_i X^{\left[i\right]}$. Then the corresponding codeword $V$ is the $\Fq$-linear span of the set $\left\{ \bigl(\al_i,f(\al_i),f(\gam \al_i), \dots, f(\gam^{s-1} \al_i) \bigr) : i \in [n]\right\}$.
 
Notice that, KK code is a special case of this for $s=1$. Since $\al_i$'s are linearly independent, each codeword is an $n$-dimensional vector space which is a subspace of
\be{ambient_space}
W=\left\langle \al_1,\al_2,\dots,\al_n \right\rangle \oplus \underbrace{\Fqm\oplus \dots \oplus \Fqm}_{s \text{ times}}
\ee
The dimension of $W$ is equal to $n+sm$, as mentioned before. Each vector in $W$ is represented as a vector with $s+1$ coordinates such as $(x,y_1,\dots,y_s)$, where $x$ is an element of the vector space spanned by $\al_1,\al_2,\dots,\al_n$ and all $y_i$'s belong to $\Fqm$.

Now, we turn to explain the list-decoding algorithm. Suppose that $V$ is transmitted and a subspace $U$ of $W$ of dimension $r$ is received. We need another parameter $d$ at the decoder which is computed as follows:
\be{define-d}
d =  \left\lceil \frac{r+s(k-1)+1}{s+1}\right\rceil
\ee
As we will see, $d$ is chosen in such a way that existence of the interpolation polynomial is guaranteed at the decoder.  
\\\textbf{List-decoding Algorithm:}
\\The decoder accepts as input a vector space $U$ which is a subspace of $W$. It then outputs a list of size at most $q^{m(s-1)}$ of vectors in $\Fqm^k$ in three steps:
\begin{enumerate}
\item \textit{Computing the interpolation points:} 
\\Find a basis $(x_i,y_{i,1},y_{i,2},\dots,y_{i,s}), i=1,2,\dots,r$, for $U$. This is the set of interpolation points.  
\item \textit{Interpolation:} Construct a nonzero multivariate linearized polynomial $Q(X,Y_1,Y_2,\dots,Y_s)$ of the form
$$
Q(X,Y)=Q_0(X)+Q_1(Y)+Q_2(Y_2)+\dots+Q_s(Y_s)
$$
where $Q_i$'s are linearized polynomials over $\Fqm$, $Q_0$ has $q$-degree at most $d-1$ and $Q_i$ has $q$-degree at most $d-k$, for $i=1,2,\dots,s$, subject to the constraint that
\be{interpolation}
Q(x_i,y_{i,1},y_{i,2},\dots,y_{i,s})=0\ \text{for}\ i=1,2,\dots,r 
\ee
\item \textit{Message recovery:} Find all polynomials $f(X) \in \L_{q^m}[X]$ of degree at most $k-1$ that satisfy the following equation
$$
Q\bigr(X,f(\gam X),f(\gam^2 X),\dots,f(\gam^{s-1}X)\bigl) = 0
$$
The decoder outputs coefficients of each solution $f(X)$ as a vector of length $k$. 
\end{enumerate}
The first step of this list-decoding algorithm can be done using elementary linear algebraic operations. The second step is basically solving a linear system of equations. There are several ways for doing that. The most straightforward way is the Gaussian elimination method. However, this method does not take advantage of the certain structure of this system of equations and therefore, it is not efficient. Efficient interpolation algorithms in the ring of linearized polynomials are presented in \cite{XYS}. In this case, the complexity of corresponding interpolation algorithm is given as $O(n^2 s^3)$ field operations over $\Fqm$. The parameter $s$ is in fact a design parameter and can be regarded as a constant. Indeed, the interpolation step is quadratic in terms of $n$. In the next subsection, we explain how the message recovery step can be done using a linear algebraic method. The complexity of the message recovery step is also quadratic. Hence, the total complexity of our algorithm is quadratic in terms of $n$, the dimension of the code.

\subsection{Recovering the Message Polynomial}
\label{sec:three_B}

As discussed in the foregoing section, in the last step of the list-decoding algorithm we need to find all polynomials $f(X) \in \L_{q^m}[X]$ of degree at most $k-1$ that satisfy
\be{factorization}
Q_0(X)+Q_1(f(X))+Q_2(f(\gam X))+\dots+Q_s(f(\gam^{s-1} X))=0
\ee
\textbf{Remark.\,} Suppose that $f, g \in \L_{q^m}[X]$ are two solutions to the equation \eq{factorization}. Since $Q_i$'s are linearized polynomials, for any $\al \in \Fq$, $\al f + (1-\al) g$ is also a solution to \eq{factorization}. Therefore, the set of solutions, which can be regarded as vectors of length $k$ over $\Fqm$, forms an affine subspace of $\Fqm^k$ as a vector space over $\Fq$. \hfill\raisebox{-0.5ex}{$\Box$}\vspace{1.0ex}

In the next lemma, we establish an upperbound on the number of solutions to \eq{factorization}. The proof of lemma also clarifies how the affine space of solutions can be computed with quadratic complexity.

\begin{lemma}
\label{list_bound}
The dimension of the affine space of solutions $f(X) \in \L_{q^m}[X]$, of degree at most $k-1$, to \eq{factorization} is at most $m(s-1)$.
\end{lemma}

\begin{proof}
For $i=0,1,2,\dots,s$, let
$$
Q_i(X)= \sum _{j \geq 0} q_{i,j} X^{q^j}
$$
If $q_{i,0}=0$ for $i=0,1,2,\dots,s$, then we replace $Q_i$ with $Q'_i$, where $Q_i(X) = Q'_i(X^q)$, in \eq{factorization} and the space of solutions remains unchanged. Therefore, one can assume that at least one $q_{i^{*},0}$ is non-zero for some $i^{*} \in \left\{0,1,2,\dots,s\right\}$. Furthermore, if $q_{1,0},q_{2,0},\dots,q_{s,0}$ are all zero, then so is $q_{0,0}$, otherwise there is no solution to \eq{factorization}. Thus, we can take $i^{*}$ from the set $\left\{1,2,\dots,s\right\}$.

Let us define the linearized polynomial $P(X)$ as
$$
P(X) = Q_0(X)+ \sum_{i=1}^{s} Q_i\bigl(f(\gam^{i-1}X)\bigr)
$$
and the polynomial $A(X)$ as
$$
A(X) = q_{1,0} + q_{2,0} X+\dots+q_{s,0} X^{s-1}
$$
Then the coefficient of $X^{q^i}$ in $P(X)$, for $i=0,1,\dots,k-1$, is equal to
\begin{align*}
q_{0,i}\ &+\  u_i\bigl( q_{1,0} + q_{2,0} \gam^{q^i}+\dots+q_{s,0} \gam ^{(s-1) q^i}\bigr)\\
&+\ u_{i-1}^q \bigl( q_{1,1} + q_{2,1} \gam^{q^i}+\dots+q_{s,1} \gam ^{(s-1) q^i}\bigr)\\
& + \dots +\  u_0^{q^i} \bigl( q_{1,i} + q_{2,i} \gam^{q^i}+\dots+q_{s,i} \gam ^{(s-1) q^i}\bigr)
\end{align*}
which can be simply expressed as
\be{i_coefficient}
q_{0,i} + A(\gam^{q^i}) u_i +  \sum^{i-1}_{j=0} a_j^{(i)} u_j^{q^{i-j}} 
\ee
for some elements $a_j^{(i)} \in \Fqm$. Now, suppose we want to find all possible solutions for $f(X)$ in \eq{factorization}. Then all the coefficients of $P(X)$ have to be equal to zero. In particular, for the coefficient of $X$ in $P(X)$: 
$$
A(\gam) u_0 + q_{0,0} =  0
$$
If $A(\gam)$ is non-zero, then $u_0 = -\frac{q_{0,0}}{A(\gam)}$. If $A(\gam)$ is zero but $q_{0,0}$ is not zero, then there is no solution for $u_0$ and consequently for $f(X)$. If both $A(\gam)$ and $q_{0,0}$ are zero, then we can set $u_0$ to any element of $\Fqm$. Then we find the solutions to $u_i$'s iteratively. For each $i$, suppose that $u_0,u_1,\dots,u_{i-1}$ are already computed. If $A(\gam^{q^i})$ is non-zero, then $u_i$ can be uniquely determined by \eq{i_coefficient}. Otherwise, we take all the elements of $\Fqm$ as possible solutions to $u_i$ and keep going for each of them separately. Notice that $A(X)$ is a non-zero polynomial of degree $s-1$ and $\gam, \gam^q, \dots, \gam^{q^{k-1}}$ are all distinct elements of $\Fqm$. Therefore, $A(\gam ^ {q^i})$ is equal to zero for at most $s-1$ possible values of $i$. This implies that the total number of solutions for $f(X)$ to \eq{factorization} is at most $q^{m(s-1)}$ which proves the lemma.
\end{proof}

\begin{corollary}
The affine space of solutions to \eq{factorization} can be computed with quadratic complexity in terms of dimension $n$. \end{corollary}


\subsection{Correctness of the Algorithm and Code Parameters}

In this subsection, we first establish the correctness of our list-decoding algorithm. Then we compute the corresponding decoding radius.

\begin{lemma}
\label{lemma_interpolation}
The particular choice of $d$ in \eq{define-d} guarantees existence of a non-zero solution for interpolation polynomial $Q$ that satisfies \eq{interpolation}.  
\end{lemma}

\begin{proof}
\eq{interpolation} defines a homogeneous system of $r$ linear equations. The number of unknown coefficients is equal to
$$
d + (d-k+1)s = d(s+1) - s(k-1)
$$
A non-zero solution for this homogeneous system of linear equations is guaranteed if the number of equations is strictly less than the number of variables. i.e.
\begin{equation*}
\begin{split}
r &\leq  d(s+1) - s(k-1) - 1  \Leftrightarrow \\
d &\geq \frac{r + s(k-1) + 1}{s+1} 
\end{split}
\end{equation*}
This is guaranteed by the choice of $d$ in \eq{define-d}.
\end{proof}

We form the following linearized polynomial $E(X)$ wherein $\fbu(X)$ is the message polynomial and $Q(X,Y_1,\dots,Y_L)$ is the interpolation polynomial provided by the list-decoding algorithm. 
\begin{align*}
E(X) &= Q\bigl(X,\fbu(X),\fbu(\gam X), \dots, \fbu(\gam^{s-1} X)\bigr) \\
&= Q_0(X) + \sum^{s}_{i=1}Q_i \otimes \fbu(\gam^{i-1} X)
\end{align*}
Let $\rho$ and $t$ denote the number of erasures and errors in the received subspace $U$, respectively. Hence, the dimension of $U$ is in fact equal to $ r = n - \rho + t$.

\begin{lemma}
\label{lemma_roots}
The linearized polynomial $E(X)$ has at least $n-\rho$ linearly independent roots in $\Fqm$.
\end{lemma}
\begin{proof}
Let $U'$ denote the intersection of the transmitted codeword $V$ and the received subspace $U$. Then $U'$ is a subspace of the received vector space $U$ with dimension $n-\rho$. Since $Q$ is a linearized polynomial
$$
Q(x,y_1,\dots,y_s)=0
$$
for any $(x,y_1,\dots,y_s) \in U'$. On the other hand, $(x,y_1,\dots,y_s)$ is also an element of the transmitted codeword $V$. Therefore, 
$$
(x,y_1,\dots,y_s) = \bigl(\beta,\fbu(\beta), \fbu(\gam\beta), \dots, \fbu(\gam^{s-1}\beta)\bigr)
$$
for some $\beta$ in the linear span of $\al_1,\al_2,\dots,\al_n$. Therefore, $\beta$ is a root for the polynomial $E(X)$. Hence, there are at least $n-\rho$ linearly independent roots for $E(X)$. 
\end{proof}

\begin{corollary}
\label{cor_roots}
 If $d \leq n-\rho$, then the linearized polynomial $E(X)$ is identically zero.
\end{corollary}
\begin{proof}
The $q$-degree of $\fbu(X)$ is at most $k-1$. Therefore, the $q$-degree of $Q_i \otimes \fbu(\gam^{i-1}X)$ is at most 
$$
d - k + k - 1 = d - 1 
$$
for $i=1,\dots,L$. Also, the $q$-degree of $Q_0(X)$ is at most $d-1$. Thus the $q$-degree of $E(X)$ is at most $d-1$. On the other hand, $E(X)$ has at least $n-\rho$ linearly independent roots by \Lref{lemma_roots}. Therefore, $E(X)$ must be the all zero polynomial.
\end{proof}

\begin{theorem}
\label{thm_correctness}
The output of our list-decoding algorithm is a list of size at most $q^{m(s-1)}$ which includes the transmitted message $\bu$ provided that
\be{error_bound}
s\rho + t < ns - s(k-1)
\ee
\end{theorem}
\begin{proof}
The existence of non-zero interpolation polynomial $Q$ that satisfies \eq{interpolation} is guaranteed by \Lref{lemma_interpolation}. Then by \Cref{cor_roots}, $E(X)$ is the all zero polynomial provided that
\be{thm_correctness_1}
\left\lceil \frac{r+s(k-1)+1}{s+1}\right\rceil \leq (n-\rho)
\ee
where we have used the expression for $d$ from \eq{define-d}. We plug in $r=n-\rho+t$ into \eq{thm_correctness_1}. Then observe that \eq{thm_correctness_1} is in fact equivalent to
$$
s\rho + t < ns - s(k-1)
$$
Thus this condition on the number of errors and erasures implies that $E(X)$ is identically zero. Therefore, the message polynomial $\fbu(X)$ is a solution to \eq{factorization}. There are at most $q^{m(s-1)}$ solutions to \eq{factorization} by \Lref{list_bound}. Therefore, the list size is at most $q^{m(s-1)}$.
\end{proof}

Now, we turn to the parameters of the proposed construction. The ambient space $W$ is given in \eq{ambient_space} which has dimension equal to $n+sm$. The symbol rate $R$ and the packet rate $\Rs$ of the code can be computed as defined in \eq{symbol_rate} and \eq{packet_rate}:
\begin{equation*}
\begin{split}
R &= \frac{\text{log}_q(\text{size of the code})}{n(\text{dim}(W))} \, = \, \frac{km}{n(n+sm)}\\
\Rs &= \frac{\text{log}_{q^m}(\text{size of the code})}{n} \, = \, \frac{k}{n}
\end{split}
\end{equation*}
The normalized decoding radius, in which erasures have weight $s$, is given by \Tref{thm_correctness} as 
\begin{equation}
\label{decoding_radius}
\begin{split}
&\, s - \frac{s(k-1)}{n}\\
& \approx \,  s \,-\,  s^2(1+\frac{n}{ms}) R
\end{split}
\end{equation}
In the regime where $n$ is much smaller than $ms$, the error decoding radius can be approximated as $s - s^2 R$. 

The normalized decoding radius in terms of the packet rate $\Rs$ can be approximated as $s(1-\Rs)$. It implies that, for any packet rate $\Rs$, we basically achieve any decoding radius by letting the list size to be large enough.

\vspace{1ex}
\section{List-decoding of Gabidulin Codes}
\vspace{.25ex}
\label{sec:four}
In this section, we first introduce a folded version of Gabidulin codes. Then, we propose a list-decoding algorithm which provides decoding radius up to the Singleton bound $1-R$, the best possible trade-off between the rate and error-correction radius.

Let $\gam$ be a primitive element of $\Fqm$. Let $\cC$ denote the Gabidulin code constructed with parameters $\al_i = \gam ^{[i-1]}$ as discussed in Section \ref{sec:two_B}. Let also $h$ be a positive integer that divides $n$ and let $g=n/h$.
\\\textbf{Definition.\,} (Folded Gabidulin Code)
\\\ The $h$-folded version of Gabidulin code $\cC$ is a code whose codewords are elements of $\Fq^{g \times hm}$. The encoding of a message polynomial $\fbu(X)$ of $q$-degree at most $k-1$ has as its $i$-th row, for $0 \leq i < g$, the $h$-tuple $\bigl(\fbu(\gam^{ih}), \fbu(\gam^{ih+1}), \dots, \fbu(\gam^{(i+1)h-1})\bigr)$,\ which can be regarded as an element in $\Fq^{hm}$. 
\hfill\raisebox{-0.5ex}{$\Box$}\vspace{1.0ex}

Notice that folding does not change the rate. The rate of folded version of code $\cC$ is equal to the rate of $\cC$ which is equal to $k/n$. 

Before going into the details of list-decoding algorithm, we would like to clarify the difference between the notion of "error" in subspace codes and rank-metric codes. Suppose that a codeword $\bX$ in code $\cC$ is transmitted and a word $\bY$ with $t$ errors is received i.e. rank$(\bX-\bY) = t$. Now consider $\left\langle \bX \right\rangle$ and $\left\langle \bY \right\rangle$ in the context of subspace codes. Then $\left\langle \bY \right\rangle$ is corrupted with $t$ errors and $t$ erasures with respect to $\left\langle \bX \right\rangle$. In fact, in rank-metric codes, there is no notion of "erasure" and each error is corresponding to one error and one erasure in the context of subspace codes.

For $0 \leq i \leq g$ and $0 \leq j \leq m$, let $y_{i,j} \in \Fqm$ denote the $(i,j)$-th coordinate of received word $\bY$ regarded as a matrix in $ \Fqm^{g \times m}$. Let $s$ be a positive integer less than or equal to $h$. We propose a decoding algorithm based on interpolating an $s+1$-variate linearized polynomial $Q(X,Y_1,\dots,Y_s)$. The $q$-degree of $Q$ is characterized in terms of parameter $d$ which is  set as follows:
\be{define_d2}
d =  \left\lceil \frac{g(h-s+1)+s(k-1)+1}{s+1}\right\rceil
\ee
This particular choice of $d$ will guarantee existence of the interpolation polynomial.
\\\textbf{List-decoding algorithm of folded Gabidulin codes\,} 

\begin{enumerate}

\item \textit{Interpolation:} Construct a nonzero multivariate linearized polynomial $Q(X,Y_1,Y_2,\dots,Y_s)$ of the form
$$
Q(X,Y)=Q_0(X)+Q_1(Y)+Q_2(Y_2)+\dots+Q_s(Y_s)
$$
where $Q_i$'s are linearized polynomials over $\Fqm$, $Q_0$ has $q$-degree at most $d-1$ and the $q$-degree of all other $Q_i$'s is at most $d-k$ subject to the constraint that
\be{interpolation2}
Q(\gam^{ih+j},y_{i,j},y_{i,j+1},\dots,y_{i,j+s-1})=0
\ee
for $i = 0,1,\dots,g-1$ and $j=0,1,\dots,h-s$.
\item \textit{Message recovery:} Find all the solutions $f(X) \in \L_{q^m}[X]$ to the following equation:
\be{factorization_2}
Q\bigr(X,f(\gam X),f(\gam^2 X),\dots,f(\gam^{s-1}X)\bigl) = 0
\ee
The decoder outputs coefficients of each solution $f(X)$ as a vector of length $k$. 
\end{enumerate}

The interpolation step is very similar to the interpolation step of the list-decoding algorithm discussed in Section \ref{sec:three_A}. It can be done using either the straightforward Gaussian elimination method or an efficient interpolation algorithm in the ring of linearized polynomials as presented in \cite{XYS}, similar to the algorithm in Section \ref{sec:three_A}. The message recovery step is exactly similar to that of list-decoding algorithm in Section \ref{sec:three_A}. It can be also done as discussed in Section \ref{sec:three_B}. The total complexity of our list-decoding algorithm is then quadratic in terms of dimension $n$.
 
Next, we establish correctness of the proposed list-decoding algorithm and compute the decoding radius of the code.

\begin{lemma}
\label{lemma_interpolation2}
The particular choice of $d$ in \eq{define_d2} guarantees existence of a non-zero solution for interpolation polynomial $Q$ that satisfies \eq{interpolation2}.  
\end{lemma}

\begin{proof}
\eq{interpolation2} is in fact a homogeneous system of $g(h-s+1)$ linear equations. The number of unknown coefficients is given by
$$
d + (d-k+1)s = d(s+1) - s(k-1)
$$
If the number of equations is strictly less than the number of variables in a homogeneous system of linear equations, then a non-zero solution is guaranteed to exist . i.e.
\begin{equation*}
\begin{split}
g(h-s+1) &\leq  d(s+1) - s(k-1) - 1  \Leftrightarrow \\
d &\geq \frac{g(h-s+1) + s(k-1) + 1}{s+1} 
\end{split}
\end{equation*}
This is guaranteed by the choice of $d$ in \eq{define_d2}.
\end{proof}
 
Let $\bU \in \Fq^{g \times hm} $ denote the codeword corresponding to the message polynomial $\fbu(X)$.  Then $\left\langle \bU \right\rangle \cap \left\langle \bY \right\rangle$, the intersection of the row spaces of matrices $\bU$ and $\bY$, has dimension $g-t$, where $t$ is the rank of error. We also define the linearized polynomial $E(X)$ as follows:
\begin{align*}
E(X) &= Q\bigl(X,\fbu(X),\fbu(\gam X), \dots, \fbu(\gam^{s-1} X)\bigr) \\
&= Q_0(X) + \sum^{s}_{i=1}Q_i \otimes \fbu(\gam^{i-1} X)
\end{align*}

\begin{lemma}
\label{lemma_roots2}
There are at least $(g-t)(h-s+1)$ linearly independent roots in $\Fqm$ for the linearized polynomial $E(X)$.
\end{lemma}
\begin{proof}
Notice that any element in the row space of $\bU$ can be represented as 
$$
\bigl( (\fbu(\beta),\fbu(\gam \beta),\dots,\fbu(\gam^{h-1} \beta) \bigr)
$$
for some $\beta \in \Fqm$. Now consider a basis for $\left\langle \bU \right\rangle \cap \left\langle \bY \right\rangle$. Indeed, the basis can be represented as 
$$
\bigl\{ \bigl( (\fbu(\beta_i),\fbu(\gam \beta_i),\dots,\fbu(\gam^{h-1} \beta_i) \bigr)\,:\, i=1,2,\dots,n-t\bigr\}
$$ 
where $\beta_1,\dots,\beta_{g-t}$ are $g-t$ linearly independent elements of $\Fqm$. In fact, they are taken from the subspace spanned by $1,\gam^h,\dots,\gam^{h(g-1)}$. Then linearity of the interpolation $Q$ and \eq{interpolation2} together imply that
$$
Q\bigl(\gam^j \beta_i, \fbu(\gam^j \beta_i), \fbu(\gam^{j+1}) \beta_i,\dots, \fbu(\gam^{j+s-1}\beta_i)\bigr) = 0
$$
for $i = 1,2,\dots,g-t$ and $j=0,1,\dots,h-s$. It is indeed equivalent to $\gam^j \beta_i$ being a root for $E(X)$. We claim that $\gam^j \beta_i$, for $i = 1,2,\dots,g-t$ and $j=0,1,\dots,h-s$ are all linearly independent elements of $\Fqm$. Let $\cP_j$ denote the subspace spanned by $\gam^j,\gam^{j+h},\dots,\gam^{j+h(g-1)}$. Since $1,\gam, \dots, \gam^{n-1}$ are all linearly independent,  $\cP_j$'s are all disjoint. Also, $\gam^j \beta_i$, for $i = 1,2,\dots,g-t$, are $g-t$ linearly independent elements of $\cP_j$. This completes the proof of claim. Therefore, $\gam^j \beta_i$, for $i = 1,2,\dots,g-t$ and $j=0,1,\dots,h-s$, are  $(g-t)(h-s+1)$ linearly independent roots for $E(X)$.   
\end{proof}

\begin{corollary}
\label{cor_roots2}
 If $d \leq (g-t)(h-s+1)$, then $E(X)$ is identically equal to zero.
\end{corollary}
\begin{proof}
The proof is very similar to the proof of \Cref{cor_roots}. The $q$-degree of $E(X)$ is at most $d-1$ by the same argument. $E(X)$ has at least $(g-t)(h-s+1)$ linearly independent roots by \Lref{lemma_roots2}. Thus, $E(X)$ must be the all zero polynomial.
\end{proof}

\begin{theorem}
\label{thm_correctness2}
If the number of errors, $t$, is bounded as
\be{t_condition}
t < \frac{gs}{s+1} \bigl( 1 - \frac{h}{h-s+1} R\bigr)
\ee
Then the proposed list-decoding algorithm of folded Gabidulin codes is correct i.e. it outputs a list of size at most $q^{m(s-1)}$ which includes the transmitted message $\bu$.
\end{theorem}
\begin{proof}
The interpolation polynomial $Q$ that satisfies \eq{interpolation2} is guaranteed to exist by \Lref{lemma_interpolation2}. If 
\be{thm_correctness_2}
\left\lceil \frac{g(h-s+1)+s(k-1)+1}{s+1}\right\rceil \leq (g-t)(h-s+1)
\ee
then by \Cref{cor_roots2} and using the expression for $d$ from \eq{define_d2}, $E(X)$ is the all zero polynomial. \eq{thm_correctness_2} is equivalent to
$$
g(h-s+1) + s(k-1) < (g-t)(h-s+1)(s+1)
$$
which can be simplified to \eq{t_condition} by using the approximation 
$$
R \approx \frac{k-1}{n}
$$
Therefore, the message polynomial $\fbu(X)$ is a solution to \eq{factorization_2}. There are at most $q^{m(s-1)}$ solutions to \eq{factorization_2} by \Lref{list_bound}. Therefore, the list size is at most $q^{m(s-1)}$.
\end{proof}

\begin{corollary}
\label{cor_radius}
The normalized decoding radius of folded Gabidulin code using the proposed list-decoding algorithm is equal to
$$
 \frac{s}{s+1} \bigl( 1 - \frac{h}{h-s+1} R\bigr)
$$
\end{corollary}

If we let both $s$ and $h$ grow large while $s$ is much smaller that $h$, we get decoding radius arbitrary close to $1-R$. Notice that $1-R$ is indeed equal to the normalized minimum rank distance of the code. This means that we are able to achieve the ultimate error-correction radius for rank-metric codes.This result is stated in the following theorem.

\begin{theorem}
\label{main_theorem}
For every $\eps > 0$ and $0 < R < 1$, there is a family of folded Gabidulin codes with rate $R$ that can be list-decoded up to normalized number of errors $1 - R - \eps$. The size of output list is at most $\cQ ^ {O(1/\eps)}$, where $\cQ$ is the size of the field that message symbols are chosen from. 
\end{theorem}  

\begin{proof}
Given $R$ and $\eps$, we can apply the results of \Tref{thm_correctness2} and \Cref{cor_radius} with the choice $s = 1 /2\eps$ and $h = 1 / 4\eps^2$.
\end{proof}

\vspace{1ex}


\bibliographystyle{IEEEtran}

\end{document}